\documentclass[11pt]{article}

\usepackage{fancyhdr,fancybox}
\usepackage{varioref,float,amssymb,amsmath}
\usepackage{float,amsfonts,amsthm}
\usepackage{changebar,longtable}
\usepackage{makeidx}
\usepackage[below]{placeins}
\usepackage{flafter}

\newcommand{\ket}[1]{\left| #1 \right\rangle} 
\newcommand{\bra}[1]{\left\langle #1 \right|} 
\newcommand{\braket}[2]{\left\langle #1 \big| #2 \right\rangle}                 
\newcommand{\bracket}[3]{\left\langle #1 \big| #2 \big| #3 \right\rangle}       

\newtheorem{theorem}{Theorem}[section]

\begin{document}

\title{Element Distinctness Revisited}
\author{Renato Portugal\footnote{portugal@lncc.br} \\ 
\\
{\small National Laboratory of Scientific Computing - LNCC} \\
{\small Av. Get\'{u}lio Vargas 333, Petr\'{o}polis, RJ, 25651-075, Brazil}
}

\maketitle
\begin{abstract}
The element distinctness problem is the problem of determining whether the elements of a list are distinct, that is, if $x=(x_1,...,x_N)$ is a list with $N$ elements, we ask whether the elements of $x$ are distinct or not. The solution in a classical computer requires $N$ queries because it uses sorting to check whether there are equal elements. In the quantum case, it is possible to solve the problem in $O(N^{2/3})$ queries. There is an extension which asks whether there are $k$ colliding elements, known as element $k$-distinctness problem. 

This work obtains optimal values of two critical parameters of Ambainis' seminal quantum algorithm [\textit{SIAM J.~Comput.}, 37, 210-239, 2007]. The first critical parameter is the number of repetitions of the algorithm's main block, which inverts the phase of the marked elements and calls a subroutine. The second parameter is the number of quantum walk steps interlaced by oracle queries. We show that, when the optimal values of the parameters are used, the algorithm's success probability is $1-O(N^{1/(k+1)})$, quickly approaching 1. The specification of the exact running time and success probability is important in practical applications of this algorithm.
\end{abstract}

\section{Introduction}

The element distinctness problem has a long history. In classical computing, the optimal lower bound for the model of comparison-based branching programs was obtained by Yao~\cite{Yao88} and classical lower bounds have been obtained in general models in Refs.~\cite{GKHS96,BSSV03}. Quantum lower bounds for the number of queries were obtained by Aaronson and Shi~\cite{AS04} and Ambainis~\cite{Amb05}. Buhrman~\textit{et al.}~\cite{BDHHMSW05} described a quantum algorithm that uses $O(N^{3/4})$ queries. Ambainis' optimal algorithm for the element distinctness problem in $O(N^{2/3})$ queries  firstly appeared in~\cite{Amb04} and later in~\cite{Amb07a}. Ambainis also addressed the element $k$-distinctness problem describing an algorithm in $O(N^{k/(k+1)})$ queries. This algorithm used a new quantum walk framework on a {bipartite graph}, which was generalized by Szegedy~\cite{Sze04a}. The algorithm was also used to build a quantum algorithm for triangle finding by Magniez~\textit{et al.}~\cite{MSS07} and to subset finding by Childs and Eisenberg~\cite{CE05}.

A related problem is the collision problem, where a one-to-one or a two-to-one function $f:\{1,...,N\}\rightarrow\{1,...,N\}$ is given and we have to decide which the function type is. Quantum lower bounds for the collision problem were obtained by Aaronson and Shi~\cite{AS04} and by Kutin~\cite{Kut05}. Brassard~\textit{et al.}~\cite{BHT98b} solved the collision problem in $O(N^{1/3})$ quantum steps achieving the lower bound. If the element distinctness problem can be solved with $N$ queries, then the collision problem can be solved with $O(\sqrt{N})$ queries~\cite{AS04}.

Many important results were obtained after Ambainis' seminal paper. Santha~\cite{San08} surveyed the application of Szegedy's quantum walk to the element distinctness problem and for other related search problems, such as matrix product verification and group commutativity. Childs \cite{Chi10} described the element distinctness algorithm in terms of the continuous-time quantum walk model~\cite{FG98}. Belovs~\cite{Bel12} used learning graphs to present a quantum algorithm for the $k$-distinctness problem in $O\big(N^{1-2^{k-2}/(2^k-1)}\big)$ queries, improving Ambainis' result for $k\ge 3$ and Belovs~\textit{et al.}~\cite{BCJKM13} presented quantum walk algorithms for the element 3-distinctness problem with time complexity $\tilde O(N^{5/7})$ improving the time complexity of $\tilde O(N^{3/4})$ by Ambainis. Rosmanis~\cite{Ros14} addressed quantum adversary lower bounds for the element distinctness problem. Kaplan~\cite{Kap16} used the element distinctness algorithm in the context of quantum attacks against iterated block ciphers. Jeffery~\textit{et al.}~\cite{JMW17} analyzed parallel quantum queries for the element distinctness problem. 

Ambainis' algorithm consists of a main block that is repeated $t_1=\linebreak O((N/r)^{k/2})$ times. The block alternates the action of a conditional phase-flip operator and a subroutine call. The subroutine executes $t_2=O(\sqrt{r})$ steps of a bipartite quantum walk interlaced with oracle calls. The value of $r$ is chosen so that the number of queries is minimized. Ambainis showed that the best value is $r=N^{k/(k+1)}$, which implies that $t_1=O(\sqrt{r})$ and the number of queries is $O(r)$, which is optimal for $k=2$.

In this work, we determine the optimal values of constants $c_1$ and $c_2$ that maximize the success probability of the algorithm, where $t_1=c_1\sqrt{r}$ and $t_2=c_2\sqrt{r}$. We show that the optimal values are $c_1=\pi/4$ and $c_2=\pi/(2\sqrt{k})$ and the success probability is $1-O(1/\sqrt[k]{r})$. In order to do so, we use an instance of the staggered quantum walk~\cite{PSFG16,Por16,Por16b}, which helps to simplify the analysis of the algorithm. The algorithm can be described as a quantum-walk-based search algorithm with multiple marked vertices. At the end, we measure the position of the walker outputting a vertex label, which is a $r$-subset of indices $\{i_1,...,i_r\}$ that has a $k$-collision with high probability, that is, $x_{j_1}=\cdots=x_{j_k}$ for $\{j_1,...,j_k\}\subset \{i_1,...,i_r\}$. The algorithm can be analyzed in full details because the dynamics can be obtained from a reduced $(2k+1)$-dimensional Hilbert space, which simplifies the analysis of the quantum-walk-based search algorithm. This work was motivated by Abreu's master thesis~\cite{Abr17}, who analyzed Ambainis' algorithm in terms of the staggered model.


In Section~\ref{Sec:companalysis}, we describe the staggered quantum walk and the graph on which the quantum walk takes place. Then, we describe the algorithm. In Section~\ref{sec:theoproof}, we formulate a theorem about the optimality of $t_1$, $t_2$, and the algorithm's success probability, and we give a proof of this theorem. In Section~\ref{sec:conc}, we draw our conclusions. The appendix includes a formal definition of the staggered quantum walk and contains a glossary of graph theory terms used in this work.

\section{Description of the Algorithm}\label{Sec:companalysis}

This section describes a quantum algorithm for the element $k$-distinctness problem, which is the following problem. Suppose we have a list $x=(x_1,...,x_N)$ of $N$ elements, is there a set $K=\{{i_1},...,{i_k}\}$ with $k$ distinct marked indices such that $x_{i_1}=\cdots = x_{i_k}$?

\subsection{Quantum Walk Evolution Operator}

Before describing the algorithm, let us give a list of definitions. $[N]$ is the set $\{1,...,N\}$, $r$ is the integer nearest to $N^{\frac{k}{k+1}}$, ${\mathcal{S}}_r$ is the set of all $r$-subsets of $[N]$, ${\mathcal{V}}=\{(S,y)\,:\,S\in{\mathcal{S}}_r,\,y\in [N]\setminus S\} $, and ${\mathcal{H}}=\textrm{span}\left\{\ket{S,y}\,:\,(S,y)\in {\mathcal{V}}\right\}$. Note that $\left|{\mathcal{S}}_r\right|={N\choose r}$ and $\left|{\mathcal{V}}\right|={N\choose {r}}(N-r)$, where ${N\choose r}$ is the {binomial coefficient}.  
Let $\Gamma$ be a graph with vertex set ${\mathcal{V}}$ (the vertices are labeled by $(S,y)\in {\mathcal{V}}$) such that vertices $(S,y)$ and $(S',y')$ are {adjacent} if and only if $S=S'$ or $S\cup \{y\}=S'\cup \{y'\}$. Since the clique graph of $\Gamma$ is 2-colorable, $\Gamma$ is 2-tessellable~\cite{Por16b}. 

The appendix contains definitions of some key concepts of graph theory and, in particular, the definitions of graph tessellation and the staggered model, which are required in the rest of this work. Note that from now on, the description given here moves away from Ambainis' description, which uses a bipartite graph defined formally in the appendix. Graph $\Gamma$, on which the staggered quantum walk is defined, is the line graph of Ambainis' graph. The results obtained in this paper regarding the optimality of $t_1$ and $t_2$ do apply to Ambainis' algorithm. Note that there is an alternate description of the element distinctness algorithm given by Santha~\cite{San08}, who uses Szegedy's quantum walk on a symmetric bipartite graph obtained from the Johnson graph through a duplication process. In principle, the results obtained here do not apply to Santha's version.

A staggered quantum walk on graph $\Gamma$ is defined after describing two tessellations ${\mathcal{T}_\alpha}$ and ${\mathcal{T}_\beta}$ induced by a coloring of the clique graph $K(\Gamma)$. Let us start by defining ${\mathcal{T}_\alpha}$. For each $S\in {\mathcal{S}}_r$ define set $\alpha_S\,=\,\{(S,y)\in {\mathcal{V}}\,:\,y\in [N]\setminus S\}$. We state that $\alpha_S$ is a {clique} of size $(N-r)$. In fact, a subset of vertices is a {clique} if all vertices in the subset are adjacent. By definition, $\alpha_S$ is a subset of vertices and all vertices in $\alpha_S$ are adjacent because they share the same $S$. The size of the clique is $(N-r)$ because the cardinality of set $[N]\setminus S$ is $(N-r)$. It is straightforward to check that the union of $\alpha_S$ for all $S$ in ${\mathcal{S}}_r$ is the vertex set ${\mathcal{V}}$, that is, 
$${\mathcal{V}}=\bigcup_{S\in {\mathcal{S}}_r}{\alpha_S}.$$
Besides, $\alpha_S\cap\alpha_{S'}=\emptyset$ if $S\neq S'$. Then, the set ${\mathcal{T}_\alpha}=\{\alpha_S\,:\,S\in {\mathcal{S}}_r\}$ is a tessellation of $\Gamma$, whose size is $\left|{\mathcal{T}_\alpha}\right|={N\choose r}$. 

For each $S\in {\mathcal{S}}_r$, define the $\alpha$-polygon vector
\begin{equation}\label{ed_ket_alpha}
	\ket{\alpha_{S}}\,=\, \frac{1}{\sqrt{N-r}}\sum_{y\in [N]\setminus S}\ket{S,y}.
\end{equation}
Note that $\braket{{\alpha_{S}}}{{\alpha_{S'}}}=\delta_{SS'}$. Now define
\begin{equation}\label{ed_U_0}
	U_{\alpha}\,=\,2\sum_{S\in {\mathcal{S}}_r}\ket{\alpha_{S}}\bra{\alpha_{S}}-I,
\end{equation} 
which is the unitary and Hermitian operator associated with tessellation $\alpha$.

Let us define tessellation ${\mathcal{T}_\beta}$. Define a partition of the vertex set ${\mathcal{V}}$ induced by the equivalence relation $\sim$, where $(S,y)\sim  (S',y')$ if and only if $S\cup \{y\}=S'\cup \{y'\}$. An equivalence class is defined by $[S,y]=\big\{(S',y')\in {\mathcal{V}}\,:\,(S',y')$ $\sim$  $(S,y)\big\}$ and the quotient set by ${\mathcal{V}}/{\sim}=\left\{[S,y]\,: \, (S,y)\in  {\mathcal{V}}\right\}$. Note that the cardinality of each equivalence class is $(r+1)$ and of the quotient set is ${N\choose {r+1}}$. For each element $[S,y]$ in the quotient set, 
 define $\beta_{[S,y]}\,=\,\{(S',y')\in {\mathcal{V}}\,:\,(S',y')\sim  (S,y)\}$. Set $\beta_{[S,y]}$ is obtained from a cyclic rotation of the elements of $S\cup\{y\}$.
We state that $\beta_{[S,y]}$ is a {clique} of size $(r+1)$. In fact, all vertices $(S',y')$ in $\beta_{[S,y]}$ are adjacent because $S'\cup \{y'\}=S\cup \{y\}$. The size of $\beta_{[S,y]}$ is $(r+1)$ because the cardinality of set $S\cup \{y\}$ is $(r+1)$.

It is straightforward to check that the union of $\beta_{[S,y]}$ for all $[S,y]$ in quotient set ${\mathcal{V}}/{\sim}$ is the vertex set ${\mathcal{V}}$, that is, 
$${\mathcal{V}}=\bigcup_{[S,y]\in{\mathcal{V}}/{\sim}}{\beta_{[S,y]}}.$$
Besides, $\beta_{[S,y]}\cap\beta_{[S',y']}=\emptyset$ if $[S,y]\neq [S',y']$. Then, the set ${\mathcal{T}_\beta}=\{\beta_{[S,y]}\,:\,{[S,y]\in{\mathcal{V}}/{\sim}}\}$ is a tessellation of $\Gamma$, whose size is $\left|{\mathcal{T}_\beta}\right|={N\choose {r+1}}$.  

For each ${[S,y]\in{\mathcal{V}}/{\sim}}$, define the $\beta$-polygon vector 
\begin{equation}\label{ed_ket_beta}
	\ket{\beta_{[S,y]}}\,=\, \frac{1}{\sqrt{r+1}}\sum_{y'\in S\cup \{y\}}\ket{S\cup \{y\}\setminus \{y'\},y'}.
\end{equation}
Note that $\ket{\beta_{[S,y]}}$ is the uniform superposition of the equivalence class that contains $(S,y)$ and that $\braket{\beta_{[S,y]}}{\beta_{[S',y']}}=\delta_{[S,y],[S',y']}$. Define
\begin{equation}\label{ed_U_1}
	U_{\beta}\,=\,2\sum_{[S,y]\in{\mathcal{V}}/{\sim}}\ket{\beta_{[S,y]}}\bra{\beta_{[S,y]}}-I,
\end{equation}
which is the unitary and Hermitian operator associated with tessellation $\beta$.

$U_\alpha$ and $U_\beta$ are local operators in the sense that they move the walker only to adjacent vertices. The evolution operator of a staggered quantum walk on graph $\Gamma$ with unmarked vertices is driven by the unitary operator $U=U_\beta U_\alpha$. The evolution operator must be modified if there are marked vertices. A vertex $(S,y)$ is marked if and only if $K\subseteq S$. The usual recipe to obtain quantum walk search algorithms is to define a new evolution operator $U'=U R$, where $R$ inverts the sign of the marked vertices and acts as the identity on unmarked ones. This recipe does not work in the present case because the argument of the principal eigenvalue of $U$ goes to zero too quickly when $N$ increases. To solve this problem, we have to use the recipe $U'=U^{t_2} R$, where $t_2$ must counteract the decrease of the argument of the principal eigenvalue.

In the next section, we give the full description of the element $k$-distinctness algorithm, which employs two registers. Here we move closer to Ambainis' description but from two key differences: First, the unitary operator $U_{\beta}$, which acts on the first register only, is extended to $U_\beta^{^{\textrm{EXT}}}$, which acts on both registers; and second, the oracle is simpler and acts only on the last ket of the second register. In Ambainis' algorithm, the operator that is equivalent to $U_{\beta}$ acts only on the first register and the oracle must perform some highly non-trivial tasks.

\subsection{The Algorithm}

The algorithm uses two registers.  A vector of the computation basis has the form
$$\ket{S,y}\otimes\ket{x'_1,...,x'_{r+1}},$$
where $(S,y)$ is a vertex label, $x'\in [M]$, and $M$ is an upper bound for the list elements. The Hilbert space has dimension ${N\choose {r}}(N-r)M^{r+1}$ and the memory in qubits is then $O\big(r(\log_2 N+\log_2 M)\big)$.  The notation $x'_i$ denotes a generic value in $[M]$ and $x_i$ denotes the list element in the $i$th position, such as $x_{i_1}$ or $x_y$.

\subsubsection*{Initial Setup} 

The initial state is
\begin{equation*}
	\frac{1}{\sqrt{{N\choose {r}}(N-r)}}\sum_{(S,y)\in {\mathcal{V}}}\ket{S,y}\ket{0,...,0}.
\end{equation*}
The first step is to query each $x_i$ for $i\in S$. Suppose that $S=\{i_1,...,i_r\}$, then the next state is
\begin{equation}\label{ed_ic_is}
	\frac{1}{\sqrt{{N\choose {r}}(N-r)}}\sum_{(S,y)\in {\mathcal{V}}}\ket{S,y}\ket{x_{i_1},...,x_{i_r},0}.
\end{equation}

\subsubsection*{Main Block}

\begin{enumerate}

\item Repeat this step the following number of times:
 $t_1= \Big\lfloor{\frac{\pi}{4} \sqrt{r}}\Big\rceil$, where $\lfloor\,\,\rceil$ is the notation for the nearest integer.

\begin{enumerate}
\item Apply a conditional phase-flip operator ${\mathcal{R}}$ that inverts the phase of $\ket{S,y}$ $\ket{x'_1,...,x'_{r+1}}$ if and only if there is a $k$-collision for distinct indices $K=\{i_1$,...,$i_k\}$ in $S$, that is,
\[
{\mathcal{R}}\ket{S,y} \ket{x'_1,...,x'_{r+1}}=
\begin{cases}
-\ket{S,y} \ket{x'_1,...,x'_{r+1}},  & \text{$k$-collision for } K\subseteq S,\\
\,\,\,\,\ket{S,y} \ket{x'_1,...,x'_{r+1}},  & \text{otherwise.}
\end{cases}
\]

\item Repeat Subroutine 1 the following number of times: 
$t_2= \Big\lfloor\frac{\pi \sqrt{r}}{2\sqrt k}\Big\rceil.$
\end{enumerate}

\item Measure the first register and check whether $S$ has a $k$-collision using a classical algorithm.
\end{enumerate}

\subsubsection*{Subroutine 1}

\begin{enumerate}
\item Apply operator $U_\alpha$ given by (\ref{ed_U_0}) on the first register.

\item Apply oracle ${\mathcal{O}}$ defined by
\begin{equation*}
	{\mathcal{O}}\ket{S,y}\ket{x'_{1},...,x'_{r+1}}=\ket{S,y}\ket{x'_{1},...,x'_{r+1}\oplus x_{y}},
\end{equation*}
which queries element $x_{y}$ and adds $x_{y}$ to $x'_{r+1}$ in the last slot of the second register.

\item Apply operator $U_\beta^{^{\textrm{EXT}}}$, which is an extension of (\ref{ed_U_1}), defined by 
\begin{equation*}\label{ed_U_1_ext}
	U_\beta^{^{\textrm{EXT}}}\,=\,2\sum_{x'_1,...,x'_{r+1}}\sum_{[S,y]\in{\mathcal{V}}/{\sim}}\ket{\beta_{[S,y]}^{x'_1,...,x'_{r+1}}}\bra{\beta_{[S,y]}^{x'_1,...,x'_{r+1}}}-I,
\end{equation*}
where
\begin{equation}\label{ed_ket_beta_ext}
	\ket{\beta_{[S,y]}^{x'_1,...,x'_{r+1}}}\,=\, \frac{1}{\sqrt{r+1}}\sum_{y'\in S\cup \{y\}}\ket{S\cup \{y\}\setminus \{y'\},y'} \ket{\pi({x'_1}),...,\pi({x'_{r+1}})}
\end{equation}
and $\pi$ is a permutation of the slots of the second register induced by the permutation of the indices of the first register.

\item Apply oracle ${\mathcal{O}}$.
\end{enumerate}

\noindent
\textbf{Notes.} (\textit{i}) In Eq.~(\ref{ed_ic_is}), the elements of $S$ and the first $r$ slots of the second register are in one-to-one correspondence. The number of queries in this step is $r$ and it is performed only once.  (\textit{ii})~When the input is state~(\ref{ed_ic_is}), the output of step~2 of Subroutine~1 has the elements of $S$ and the first $r$ slots of the second register in one-to-one correspondence and the last slot of the second register is $x_y$. This one-to-one correspondence is maintained for each term in sum (\ref{ed_ket_beta_ext}) and $\pi(x'_{r+1})=x_{y'}$. (\textit{iii})~The total number of quantum queries is $r+\pi^2r/(4\sqrt{k})$ approximately considering the Initial Setup and the Main Block. After the measurement, $r$ classical queries are necessary.

\section{Main Result}\label{sec:theoproof}

The next theorem improves the value $t_2=\pi\sqrt{r}/(3\sqrt{k})$ given by Ambainis~\cite{Amb07a}, which yields a success probability of 75\% asymptotically.

\begin{theorem} The values $t_1=\pi\sqrt{r}/4$ and $t_2=\pi\sqrt{r}/(2\sqrt{k})$ are asymptotically optimal and the success probability of the algorithm is $1-O(1/r^{1/k})$.
\end{theorem}
\begin{proof} 
Define $(2k+1)$ nonempty sets  $\eta_\ell^j=\{(S,y)\in {\mathcal{V}}:\left|S\cap K \right|=\ell, \,\left|\{y\}\cap K \right|=j \}$. Set $\eta_\ell^j$ is the set of vertices $(S,y)$ such that $S$ has exactly $\ell$ marked indices and $y\not\in K$ if $j=0$ and $y\in K$ if $j=1$. Set $\eta_k^1$ is the empty set. The cardinality of  $\eta_\ell^j$ is ${k\choose {\ell}}{N-k\choose r-{\ell}} \left( N-r-k+{\ell} \right)$ if $j=0$ and ${k\choose {\ell}}{N-k\choose r-{\ell}} \left( k-{\ell} \right)$ if $j=1$. The set of sets $\eta_\ell^j$ is a partition of ${\mathcal{V}}$. The range of $\ell$ is $0\le\ell \le k$ and of $j$ is $0 \le j \le 1$ but must exclude the case $\ell=k$ and $j=1$. Assume that $k<N-r$ so that sets $\eta_\ell^j$ are nonempty.

Define the corresponding unit vectors
\begin{equation}\label{ed_eta_j}
	\ket{\eta_{\ell}^j}\,=\,\frac{1}{\sqrt{\left|\eta_{\ell}^j\right|}}\sum_{({S,y})\in\eta_{\ell}^{j}}\ket{S,y},
\end{equation}
which span a $(2k+1)$-dimensional subspace of the Hilbert space that can be used to analyze the algorithm and to obtain the success probability. In order to do so, we show that the $(2k+1)$-subspace is invariant under the action of $U_\alpha$ and $U_\beta$. Let us obtain matrices $u_\alpha$ and $u_\beta$ of dimension $(2k+1)$ that reproduce the action of $U_\alpha$ and $U_\beta$ on vectors $\ket{\eta_{\ell}^j}$, that is,
\begin{equation*}\label{ed_Uketeta}
U_\alpha\ket{\eta_{\ell}^{j}}\,=\,\sum_{\ell'\,j'}\bracket{\ell',j'}{u_\alpha}{\ell, j}\ket{\eta_{\ell'}^{j'}},
\end{equation*}
and the equivalent one for $u_\beta$, where the set of kets $\ket{\ell,j}$ is the computational basis of a Hilbert space of dimension $(2k+1)$. 

Using~(\ref{ed_ket_alpha}) and~(\ref{ed_eta_j}), we obtain
$$\braket{\alpha_S}{\eta_{\ell}^{j}}=\frac{(1-j)(N-r)+(2j-1)(k-\ell)}{\sqrt{\left|\eta_\ell^j\right|}\sqrt{N-r}}\,\delta_{ |S\cap K|,\ell}$$
and
$$\sum_{\substack{S\in {\mathcal{S}_r}\\ |S\cap K|=\ell}}\ket{\alpha_S}=\frac{1}{\sqrt{N-r}}\left(\sqrt{\left|\eta_\ell^0\right|}\,\ket{\eta_\ell^0} + (1-\delta_{k\ell}) \sqrt{\left|\eta_\ell^1\right|}\,\ket{\eta_\ell^1}\right).$$
Using those results and~(\ref{ed_U_0}), we find the entries of $u_\alpha$, which are
\begin{align}
\bracket{\ell',j'}{u_\alpha}{\ell,j}  \,=\,\, & (-1)^j\left(1-\frac{2\,(k-l)}{{N-r}}\right)\delta_{\ell{\ell'}}\delta_{j{j'}} + \nonumber\\
& 2\,\sqrt{\frac{k-\ell}{N-r}}\sqrt{1-\frac{k-\ell}{N-r}}\,\,\delta_{\ell{\ell'}}\delta_{j\oplus 1,{j'}}.
\end{align}
Analogously, using~(\ref{ed_ket_beta}) and~(\ref{ed_eta_j}), we obtain 
$$\braket{\beta_{[S,y]}}{\eta_{\ell}^{j}}=\frac{(1-j)\,r+(2j-1)\,\ell+1}{\sqrt{\left|\eta_\ell^j\right|}\sqrt{r+1}}\,\delta_{ |(S\cup \{y\})\cap K|,\ell+j}$$
and
$$\sum_{\substack{(S,y)\\ |(S\cup \{y\})\cap K|=\ell}}\ket{\beta_{[S,y]}} = \frac{1}{\sqrt{r+1}}\left(\sqrt{\left|\eta_\ell^0\right|}\,\ket{\eta_\ell^0} + (1-\delta_{\ell 0}) \sqrt{\left|\eta_{\ell-1}^1\right|}\,\ket{\eta_{\ell-1}^1}\right). $$
Using those results and~(\ref{ed_U_1}), we find the entries of $u_\beta$, which are
\begin{align}
\bracket{\ell',j'}{u_\beta}{\ell,j}  \,=\,\, & (-1)^j\left(1-\frac{2\,(\ell+j)}{r+1}\right)\delta_{\ell\ell'}\delta_{jj'} +\nonumber\\ 
&2\,\sqrt{\frac{\ell+j}{r+1}}\,\sqrt{1-\frac{\ell+j}{r+1}}\,\,\delta_{\ell-(-1)^j,\ell'}\delta_{1\oplus j,j'} .
\end{align}

Next step is to show that the conditional phase flip operator ${\mathcal{R}}$ leaves the $(2k+1)$-subspace invariant too. ${\mathcal{R}}$ inverts the phase of $\ket{S,y}$ if and only if $|S\cap K|=k$, that is $(S,y)\in \eta_k^0$. Define a reduced version of ${\mathcal{R}}$, denoted by $R$, in the $(2k+1)$-dimensional Hilbert space as 
\begin{equation}
R=I - 2\ket{k,0}\bra{k,0}.
\end{equation}

State (\ref{ed_ic_is}) at the beginning of the algorithm is a linear combination of $\ket{\eta_{\ell}^j}$ and can be written in the Hilbert space of dimension $(2k+1)$ as
\begin{equation}\label{red_ini_state}
\ket{\psi_0}=\frac{1}{\sqrt{{N\choose {r}}(N-r)}}\sum_{\ell,j}\sqrt{\left|\eta_{\ell}^j\right|}\ket{\ell,j}.
\end{equation}
Since all steps of the algorithm can be obtained from the reduced Hilbert space, the final state of the algorithm right before the measurement can be obtained from
$$\ket{\psi_f}=\left((u_\beta u_\alpha)^{t_2} R\right)^{t_1}\ket{\psi_0}.$$
Note that the oracle changes only the second register and is therefore omitted. Our goal now is to show that the choices for $t_1$ and $t_2$ described in the algorithm are optimal with maximal success probability.

The probability of finding a marked vertex as a function of $t$ is 
\begin{equation}\label{ed_p_succ_t}
p(t)=\left|\bracket{k,0}{\left(u^{t_2}R\right)^{t}}{\psi_0}\right|^2,
\end{equation}
where $u=u_\alpha u_\beta$. Let $e^{\pm i\lambda}$ be the eigenvalues of $\left(u^{t_2}R\right)$ that are nearest to 1 and let $\ket{\lambda}$ and its complex-conjugate $\ket{\lambda}^*$ be the corresponding eigenvectors. Assume for now that the contribution of the other eigenvalues and eigenvectors to the calculation of $p(t)$ goes to zero when $N$ increases. In this case
\begin{equation}\label{ed_p_succ_t_epsionl}
p(t)=\left|e^{i\lambda\, t}\braket{k,0}{\lambda}\braket{\lambda}{\psi_0}+e^{-i\lambda\, t}\braket{k,0}{\lambda}^*\braket{\lambda}{\psi_0}^* +\epsilon\,\right|^2,
\end{equation}
where $\lim_{N\rightarrow \infty}|\epsilon|=0$. Suppose that vectors $\ket{\psi_{\pm n}}$ for $0\le n \le k$ are unit eigenvectors of $u$ with eigenvalues  $e^{ i\phi_{\pm n}}$, where $\ket{\psi_{-n}}=\ket{\psi_{n}}^*$, $\phi_{-n}=-\phi_n$, and $\braket{k,0}{\psi_{\pm n}}>0$. $\ket{\psi_0}$ is given by~(\ref{red_ini_state}) and has positive entries.

Using $\bracket{\psi_{n}}{u^{t_2} R}{\lambda}=e^{ i\lambda}\braket{\psi_{n}}{\lambda}$, we obtain  
\begin{equation}\label{psi_lambda}
 \braket{\psi_{n}}{\lambda}\,=\,\frac{2\braket{k,0}{\lambda}\braket{k,0}{\psi_{n}}}{1-{e}^{{i}(\lambda-{t_2 \phi_{n}})}},
\end{equation}
which is valid if $\lambda\neq {t_2 \phi_{n}}$. Substituting this result in
\begin{equation}
\braket{k,0}{\lambda}\,=\,\sum_{n=-k}^k \braket{k,0}{\psi_{n}}\braket{\psi_{n}}{\lambda},
\end{equation}
we obtain
\begin{equation*}
 \sum_{n} \frac{2\braket{k,0}{\psi_n}^2}{1-{e}^{{i}(\lambda-{t_2 \phi_{n}})}} =1.
\end{equation*}
Using that $2/(1-{e}^{{i}a})=1+i\sin a/(1-\cos a)$, the imaginary part of the above equation is
\begin{equation}\label{exactsum_kl}
 \sum_{n} \braket{k,0}{\psi_n}^2\frac{\sin({\lambda-{t_2 \phi_{n}}})}{1-\cos(\lambda-{t_2 \phi_{n}})} =0.
\end{equation} 
Suppose that $\lambda\ll t_2 \phi_{n}$ for $n>0$ when $N\gg 1$. We will check the validity of this assumption later. 
Expanding in Taylor series and discarding terms $O(\lambda^2)$, we obtain
\begin{equation}\label{eq:lambdaequation}
	 \lambda=\frac{\braket{k,0}{\psi_0}}{\sqrt{b}},
\end{equation}
where
\begin{eqnarray}\label{eq:a} 
		b = \sum_{n=1}^k\frac{\braket{k,0}{\psi_{n}}^2}{1-\cos(t_2\phi_n)}.
\end{eqnarray}

Now let us find $\phi_n$ and $\braket{k,0}{\psi_{n}}$. Using the entries of $u_\alpha$ and $u_\beta$, we obtain the characteristic polynomial of $u=u_\beta u_\alpha$, which is
$$\frac{|\lambda I-u|}{(r+1)^k(N-r)^k}=(\lambda-1)\prod_{n=1}^k \left(\lambda^2-2\,\lambda\cos\phi_n+1\right),$$
where
\begin{equation}\label{cosphi}
\cos\phi_n=1-{\frac {2\,n\left(N - n+1 \right)}{ \left( r+1 \right)  \left( N-r \right) }}
.
\end{equation}
The eigenvectors of $u$ can be found explicitly, but it is easier to calculate directly $\braket{k,0}{\psi_{n}}$, which for $n=0$ is
\begin{equation}\label{braketkpsi0}
	\braket{k,0}{\psi_0}\,=\,\prod _{i=0}^{k-1}\sqrt{{\frac {r-i}{N-i}}},
\end{equation}
for $0<n< k$ is
\begin{equation}\label{braketkpsik}
	\braket{k,0}{\psi_n}\,=\,\frac{1}{\sqrt{2}}\sqrt{{k\choose n}}\,\sqrt{{\frac {\prod _{i=0}^{n-1}(N-r-i)\prod _{i=n}^{k-1}(r-i)}{\prod _{i=n-1}^{2\,n-2}(N-i)\prod _{i=2\,n}^{k+n-1}(N-i)}}},
\end{equation}
and for $n=k$ is
\begin{equation}\label{braketkpsin}
	\braket{k,0}{\psi_k}\,=\,\frac{1}{\sqrt{2}}\sqrt{{\frac {\prod _{i=0}^{k-1}(N-r-i)}{\prod _{i=k-1}^{2\,k-2}(N-i)}}}.
\end{equation}

Next step in the calculation of $p(t)$ is the term $\braket{k,0}{\lambda}$. Substituting (\ref{psi_lambda}) into $\sum_{n} \left|\braket{\psi_n}{\lambda}\right|^2=1$, we have
$$\frac{1}{\left|\braket{k,0}{\lambda}\right|^2}=4\sum_{n=-k}^k\frac{\braket{k,0}{\psi_{n}}^2}{\left|1-{e}^{{i}(\lambda-{t_2 \phi_{n}})}\right|^2}.$$
Using identity $\left|1-e^{ia}\right|^2=2\,(1-\cos a)$, redefining $\ket{\lambda}$ by choosing a multiplicative unit constant such that $\braket{k,0}{\lambda}>0$, expanding in Taylor series, and keeping the dominant term, we obtain
\begin{equation}\label{braketklambda}
\braket{k,0}{\lambda}=\braket{k,0}{\lambda}^*=\frac{1}{2\sqrt{2}\sqrt{b}}.
\end{equation}
Next term is $\braket{\lambda}{\psi_0}$. Substituting $n=0$ into (\ref{psi_lambda}), using (\ref{braketklambda}), (\ref{braketkpsi0}), and (\ref{eq:lambdaequation}), we obtain 
\begin{equation}\label{braketpsi0lambda}
 \braket{\psi_0}{\lambda}=\frac{i}{\sqrt{2}} + O\left(1/\sqrt{r}\right).
\end{equation}
In fact, the next term in the asymptotic expansion of $\braket{\psi_0}{\lambda}$ would be $\lambda/\sqrt{8}$, which is $O\left(1/\sqrt{r}\right)$. Substituting (\ref{braketklambda}) and (\ref{braketpsi0lambda}) into $p(t)$, we obtain
\begin{equation}\label{p_of_t}
p(t) \,=\, \frac{1}{4\,b}\,\sin^2 \lambda t.
\end{equation}

The maximal value of the success probability is obtained by taking $t=\pi/(2\lambda)$, which implies that $p_{\textrm{succ}}=1/(4b)$. To maximize the success probability, we have to minimize $b$. Using Eq.~(\ref{eq:a}), we see that the only free parameter in $b$ is $t_2$. So, we choose $t_2$ that minimizes $b$. Eqs.~(\ref{braketkpsi0}) to~(\ref{braketkpsin}) show that $\braket{k,0}{\psi_n}^2=O\left((r/N)^{k-n}\right)$, which goes to 0 when $n<k$, and  $\braket{k,0}{\psi_k}^2=O(1)$. The expression of $b$ given by (\ref{eq:a}) is a sum of terms which is dominated by the term
$$\frac{\braket{k,0}{\psi_{k}}^2}{1-\cos(t_2\phi_k)}.$$
The optimal value of $t_2$ is the one that minimizes this term, that is, maximizes $(1-\cos(t_2\phi_k))$, which implies $t_2={\pi}/{\phi_k}$. Using (\ref{cosphi}), the asymptotic expansion of ${\pi}/{\phi_k}$ yields
\begin{equation}\label{finalt2}
t_2=\frac{\pi}{\phi_k}=\frac{\pi\sqrt{r}}{2\sqrt{k}}+O(1),
\end{equation}
and then the success probability is
\begin{equation}
	p_\textrm{succ} = 1 - \frac{k}{r^{\frac{1}{k}}}\cot^2\left(\frac{\pi}{2}\sqrt{\frac{k-1}{k}}\right) + O\left(r^{-\frac{2}{k}}\right).
\end{equation}
Now let us find the running time. Eq.~(\ref{p_of_t}) shows that the probability as a function of time is the square of a sinusoidal function. The optimal running time is the first value of $t$ that maximizes $\sin(\lambda t)$, which is $t_1=\pi/(2\lambda)$. Using the value of $t_2$ given by (\ref{finalt2}), the expression~(\ref{eq:lambdaequation}) for $\lambda$, and the dominant term of $b$, the asymptotic expansion of $\pi/(2\lambda)$ yields
\begin{equation}
t_1 \,=\, \frac{\pi}{2\,\lambda} \,=\, \frac{\pi}{4}\, \sqrt{r} + O\left(r^{\frac{k-2}{2k}}\right).
\end{equation}

It is still missing to check that $\lambda\ll t_2\phi_n$ and $\lim_{N\rightarrow \infty}|\epsilon|=0$ (see (\ref{ed_p_succ_t_epsionl})). Using (\ref{cosphi}), we have $\phi_n=O(1/\sqrt{r})$. Using (\ref{finalt2}), we obtain $t_2\phi_n=O(1)$. Then, $\lambda\ll t_2\phi_n$  when $N$ is large. On the other hand, (\ref{braketpsi0lambda}) implies that
$$\left|\braket{\psi_0}{\lambda}\right|^2+\left|\braket{\psi_0}{\lambda}^*\right|^2=1+ O\left(1/\sqrt{r}\right).$$
This means that the initial state $\ket{\psi_0}$ lies in the subspace spanned by $\ket{\lambda}$ and $\ket{\lambda}^*$ in the limit $N\rightarrow\infty$, that is, $\lim_{N\rightarrow \infty}|\epsilon|=0$. This completes the proof of the theorem.
\end{proof}

\section{Conclusions}\label{sec:conc}

In this work, we have obtained the optimal values of the two critical parameters of Ambainis' algorithm~\cite{Amb07a}. The first parameter is the number of repetitions of the steps of the main block, which includes the conditional phase-flip inversion and the subroutine call. The second parameter is the number of quantum walk steps interlaced by oracle queries. After obtaining the optimal values of the critical parameters, which are $t_1=\pi\sqrt{r}/4$ and $t_2=\pi\sqrt{r}/(2\sqrt{k})$, we have shown that the success probability of the algorithm is $1-O(1/r^{1/k})$ improving Ambainis' result, which attains a success probability of $3/4-O(1/r^{1/k})$, due to a non-optimal choice of the second critical parameter. The total number of quantum queries with optimal values for the parameters is $\big(\pi^2/(4\sqrt{k})+1\big)\,r+O\big(r^{(k-1)/k}\big)$ and, at the end, $r$ classical queries are necessary.


The dynamics of the algorithm can be described in a reduced $(2k+1)$-dimensional Hilbert space. Using the staggered quantum walk model, we were able to obtain the reduced version of the quantum-walk evolution operator and to calculate the spectral decomposition. The analysis of quantum-walk-based search algorithms using a reduced version of the evolution operator has been widely used in literature, such as Refs.~\cite{SKW03,AKR05,Tul12}, which developed many calculation tools that were used in this work and were reviewed in Ref.~\cite{Por13}.

As a follow-up, we are interested in describing quantum circuits for the element distinctness algorithm and in determining precisely the number of gates in order to find a prefactor for the known time complexity bound of $O(N^{2/3}\ln N)$~\cite{Amb07a}.

\section*{Appendix}\label{appendixA}

In this appendix, we define the graph theory terms used in this work~\cite{Die12,BLS99,Har69} and the staggered quantum walk~\cite{PSFG16,Por16,Por16b}.

A \textit{simple and undirected graph} (simply \textit{graph}) $\Gamma(V,E)$ is defined by a set $V(\Gamma)$ of vertices or nodes and a set $E(\Gamma)$ of edges so that each edge links two vertices and two vertices are linked by at most one edge. Two vertices linked by an edge are called \textit{adjacent}. Two edges  that share a common vertex are also called adjacent. A subgraph $\Gamma'(V',E')$, where $V'(\Gamma')\subset V(\Gamma)$ and $E'(\Gamma')\subset E(\Gamma)$, is an \textit{induced subgraph} of $\Gamma(V,E)$ if it has exactly the edges that appear in $\Gamma$ over the same vertex set. If two vertices are adjacent in $\Gamma$ they are also adjacent in the induced subgraph. 
A \textit{bipartite graph} is a graph whose vertex set $V$ is the union of two disjoint sets $X$ and $X'$ so that no two vertices in $X$ are adjacent and no two vertices in $X'$ are adjacent. 
A \textit{clique} is a subset of vertices of a graph such that its induced subgraph is complete. A \textit{maximal clique} is a clique that cannot be extended by including one more adjacent vertex, that is, it is not contained in a larger clique. A clique of size $d$ is called a $d$-\textit{clique}. A clique can have one vertex. A \textit{partition of the vertex set into cliques} is a collection of vertex-disjoint cliques, whose union is the vertex set. Some references in graph theory use the term ``clique'' as a synonym of \textit{maximal clique}. We avoid this notation here.
A \textit{clique graph} $K(\Gamma)$ of a graph $\Gamma$ is a graph such that every vertex represents a maximal clique of $\Gamma$ and two vertices of $K(\Gamma)$ are adjacent if and only if the underlying maximal cliques in $\Gamma$ share at least one vertex in common.
A \textit{line graph} (or \textit{derived graph} or \textit{interchange graph}) of a graph $\Gamma$ (called \textit{root graph}) is another graph $L(\Gamma)$ so that each vertex of $L(\Gamma)$ represents an edge of $\Gamma$ and two vertices of $L(\Gamma)$ are adjacent if and only if their corresponding edges share a common vertex in $\Gamma$. 
A \textit{proper coloring} or simply \textit{coloring} of a loopless graph is a labeling of the vertices with colors such that no two vertices sharing the same edge have the same color. A $n$-\textit{colorable} graph is the one whose vertices can be colored with at most $n$ colors so that no two adjacent vertices share the same color. This concept can be used for edges and other graph structures.

A \textit{graph tessellation} ${\mathcal{T}}$ is a partition of the vertex set into cliques, that is, there are disjoint cliques $c_1$, ..., $c_{\left|{\mathcal{T}}\right|}$ such $\cup_{\ell=1}^{\left|{\mathcal{T}}\right|} c_\ell=V(\Gamma)$, where $\left|{\mathcal{T}}\right|$ is the tessellation size. An element $c_\ell$ of the tessellation is called a \textit{polygon} (or \textit{tile}). An edge \textit{belongs} to the tessellation ${\mathcal{T}}$ if and only if its endpoints belong to the same polygon in ${\mathcal{T}}$. 
The set of edges belonging to ${\mathcal{T}}$ is denoted by ${\mathcal{E}}({\mathcal{T}})$.
A \textit{graph tessellation cover} of size~$n$ is a set of $n$~tessellations ${\mathcal{T}}_1,...,{\mathcal{T}}_n$, whose union is the edge set, that is, $\cup_{j=1}^n\,{\mathcal{E}}({\mathcal{T}}_j)=E(\Gamma)$. A graph is called $n$\textit{-tessellable} if there is a tessellation cover of size at most $n$. The \textit{tessellation cover number} is the size of a smallest tessellation cover of $\Gamma$. 
A graph $\Gamma$ is 2-tessellable if and only if $K(\Gamma)$ is 2-colorable~\cite{Por16b}. The definition of \textit{graph tessellation cover} was introduced by Portugal \textit{et al.}~\cite{PSFG16}. 

In its simplest form, a \textit{staggered quantum walk} on a graph $\Gamma(V,E)$ with a \textit{graph tessellation cover} ${\mathcal{C}}=\{{\mathcal{T}}_1,...,{\mathcal{T}}_n\}$ is a quantum walk driven by the unitary operator $U_{\mathcal{C}}=U_n\cdots U_2\cdot U_1$, where $U_j$ is associated with tessellation ${\mathcal{T}}_j$ and is defined by
$$U_j\,=\,2\sum_{\ell=1}^{\left|{\mathcal{T}}_j\right|} \ket{c_\ell^{(j)}}\bra{c_\ell^{(j)}} - I,$$
where
$$ \ket{c_\ell^{(j)}} \,=\, \frac{1}{\sqrt{\left|c_\ell^{(j)}\right|}}\sum_{v\in c_\ell^{(j)}}\ket{v}.$$
The dimension of Hilbert space is $|V|$ and the computational basis is indexed by the vertices of $\Gamma$.

\textit{Ambainis' graph}~\cite{Amb07a} is a bipartite graph with ${N\choose r}+{N\choose {r+1}}$ vertices. The vertices of the first set are $r$-subsets of $[N]$ and of the second set are $(r+1)$-subsets. A vertex $v_1$ in the first set is adjacent to a vertex $v_2$ in the second set if and only if $|v_1\cap v_2|=r$. The graph $\Gamma$ defined in Section~\ref{Sec:companalysis}, on which the 2-tessellable staggered quantum walk takes place, is the line graph of Ambainis' graph, which on the other hand is the clique graph of $\Gamma$. $K(\Gamma)$ is 2-colorable because Ambainis' graph is bipartite.

\section*{Acknowledgements}
The author acknowledges financial support from CNPq and thanks Raqueline Santos for useful comments.




\begin{thebibliography}{0}

\bibitem{Yao88}
A.~C.~C. Yao.
\newblock Near-optimal time-space tradeoff for element distinctness.
\newblock In {\em Proc. of 29th Annual Symposium on Foundations of
  Computer Science}, pages 91--97, 1988.
	
\bibitem{GKHS96}
D. Grigoriev, M. Karpinski, F.~M. Heide, and R. Smolensky.
\newblock A lower bound for randomized algebraic decision trees.
\newblock {\em Computational Complexity}, 6(4):357--375, 1996.

\bibitem{BSSV03}
P. Beame, M. Saks, X. Sun, and E. Vee.
\newblock Time-space trade-off lower bounds for randomized computation of
  decision problems.
\newblock {\em J. ACM}, 50(2):154--195, 2003.

\bibitem{AS04}
S. Aaronson and Y. Shi.
\newblock Quantum lower bounds for the collision and the element distinctness
  problems.
\newblock {\em J. ACM}, 51(4):595--605, 2004.

\bibitem{Amb05}
A. Ambainis.
\newblock Polynomial degree and lower bounds in quantum complexity: Collision
  and element distinctness with small range.
\newblock {\em Theory of Computing}, 1:37--46, 2005.

\bibitem{BDHHMSW05}
H.~Buhrman, C.~D\"{u}rr, M.~Heiligman, P.~H{\o}yer, F.~Magniez, M.~Santha, and
  {R. de} Wolf.
\newblock Quantum algorithms for element distinctness.
\newblock {\em SIAM Journal on Computing}, 34(6):1324--1330, 2005.

\bibitem{Amb04}
A. Ambainis.
\newblock Quantum walk algorithm for element distinctness.
\newblock In {\em FOCS '04: Proc. of the 45th Annual IEEE Symposium on
  Foundations of Computer Science}, pages 22--31, Washington, DC, 2004.

\bibitem{Amb07a}
A. Ambainis.
\newblock Quantum walk algorithm for element distinctness.
\newblock {\em SIAM Journal on Computing}, 37(1):210--239, 2007.

\bibitem{Sze04a}
M. Szegedy.
\newblock Quantum speed-up of markov chain based algorithms.
\newblock In {\em Proc. of the 45th Annual IEEE Symposium on Foundations
  of Computer Science}, FOCS '04, pages 32--41, Washington, DC, 2004. 

\bibitem{MSS07}
F. Magniez, M. Santha, and M. Szegedy.
\newblock Quantum algorithms for the triangle problem.
\newblock {\em SIAM Journal on Computing}, 37(2):413--424, 2007.

\bibitem{CE05}
A.~M. Childs and J.~M. Eisenberg.
\newblock Quantum algorithms for subset finding.
\newblock {\em Quantum Info. Comput.}, 5(7):593--604, 2005.

\bibitem{Kut05}
S. Kutin.
\newblock Quantum lower bound for the collision problem with small range.
\newblock {\em Theory of Computing}, 1:29--36, 2005.

\bibitem{BHT98b}
G. Brassard, P. H{\o}yer, and A. Tapp.
\newblock Quantum cryptanalysis of hash and claw-free functions.
\newblock In {\em Proc. of  LATIN'98: Theoretical Informatics: Third Latin American
  Symposium}, Campinas, pages 163--169, 1998.

\bibitem{San08}
M. Santha.
\newblock Quantum walk based search algorithms.
\newblock In {\em Proc. of  Theory and Applications of Models of Computation: 5th
  International Conference, TAMC 2008, Xi'an}, pages 31--46, 2008.

\bibitem{Chi10}
A.~M. Childs.
\newblock On the relationship between continuous- and discrete-time quantum
  walk.
\newblock {\em Communications in Mathematical Physics}, 294(2):581--603, 2010.

\bibitem{FG98}
E.~Farhi and S.~Gutmann.
\newblock Quantum computation and decision trees.
\newblock {\em Phys. Rev. A}, 58:915--928, 1998.

\bibitem{Bel12}
A.~Belovs.
\newblock Learning-graph-based quantum algorithm for $k$-distinctness.
\newblock In Proc. of  {\em 2012 IEEE 53rd Annual Symposium on Foundations of Computer
  Science}, pages 207--216, 2012.

\bibitem{BCJKM13}
A.~Belovs, A.~M. Childs, S.~Jeffery, R.~Kothari, and F.~Magniez.
\newblock Time-efficient quantum walks for 3-distinctness.
\newblock In {\em Proc. of  Automata, Languages, and Programming: 40th International
  Colloquium, ICALP 2013}, Riga,  pages 105--122, 2013.

\bibitem{Ros14}
A. Rosmanis.
\newblock Quantum adversary lower bound for element distinctness with small
  range.
\newblock {\em Chicago Journal of Theoretical Computer Science}, 2014(4), 2014.

\bibitem{Kap16}
M.~Kaplan.
\newblock Quantum attacks against iterated block ciphers.
\newblock {\em Mat. Vopr. Kriptogr.}, 7:71--90, 2016.

\bibitem{JMW17}
S. Jeffery, F. Magniez, and R. de~Wolf.
\newblock Optimal parallel quantum query algorithms.
\newblock {\em Algorithmica}, 79(2):509--529, 2017.

\bibitem{PSFG16}
R.~Portugal, R.~A.~M. Santos, T.~D. Fernandes, and D.~N. Gon{\c{c}}alves.
\newblock The staggered quantum walk model.
\newblock {\em Quantum Information Processing}, 15(1):85--101, 2016.

\bibitem{Por16}
R. Portugal.
\newblock Establishing the equivalence between {S}zegedy's and coined quantum
  walks using the staggered model.
\newblock {\em Quantum Information Processing}, 15(4):1387--1409, 2016.

\bibitem{Por16b}
R. Portugal.
\newblock Staggered quantum walks on graphs.
\newblock {\em Phys. Rev. A}, 93:062335, 2016.

\bibitem{Abr17}
Alexandre~S.~Abreu.
\newblock Tessela\c c\~oes em grafos e suas aplica\c c\~oes em computa\c c\~ao
  qu\^antica.
\newblock Master's thesis, UFRJ, 2017.

\bibitem{SKW03}
N.~Shenvi, J.~Kempe, and K. B.~Whaley.
\newblock A quantum random walk search algorithm.
\newblock {\em Phys. Rev. A}, 67(5):052307, 2003.

\bibitem{AKR05}
A. Ambainis, J. Kempe, and A. Rivosh.
\newblock Coins make quantum walks faster.
\newblock In {\em Proc. Sixteenth Annual ACM-SIAM Symposium on
  Discrete Algorithms, SODA}, pages 1099--1108, 2005.

\bibitem{Tul12}
A. Tulsi.
\newblock General framework for quantum search algorithms.
\newblock {\em Phys. Rev. A}, 86:042331, 2012.

\bibitem{Por13}
Renato Portugal.
\newblock {\em Quantum Walks and Search Algorithms}.
\newblock Springer, New York, 2013.

\bibitem{Die12}
Douglas B. West.
\newblock {\em Introduction to Graph Theory}.
\newblock Prentice Hall, 2000.
  

\bibitem{BLS99}
A. Brandst\"{a}dt, V.~B. Le, and J.~P. Spinrad.
\newblock {\em Graph Classes: A Survey}.
\newblock SIAM, Philadelphia, 1999.

\bibitem{Har69}
Frank Harary.
\newblock {\em Graph Theory}.
\newblock Addison-Wesley, Massachusetts, 1969.

\end{thebibliography}

\end{document}